\documentclass[letterpaper, 10 pt, conference]{ieeeconf}  % Comment this line out if you need a4paper

\IEEEoverridecommandlockouts                              % This command is only needed if 
\usepackage{amsmath,amssymb,amsfonts}
\usepackage{tabularx}
\usepackage{array}
\usepackage[caption=false,font=normalsize,labelfont=sf,textfont=sf]{subfig}
\usepackage{textcomp}
\usepackage{stfloats}
\usepackage{url}
\usepackage{verbatim}
\usepackage{graphicx}
\usepackage{cite}
\usepackage{bm}
\usepackage{subcaption} 
\usepackage{color}
\usepackage{dsfont}

\usepackage{amsmath}
\usepackage[ruled,vlined,linesnumbered]{algorithm2e}
\SetKwComment{Comment}{// }{}
\newtheorem{theo}{Theorem}

\newtheorem{assumption}{Assumption}

\newtheorem{lemma}{Lemma}

\usepackage{lscape}

\usepackage{booktabs}
\usepackage{multirow}
\title{\LARGE \bf
MM-LMPC: Multi-Modal Learning Model Predictive Control via Bandit-Based Mode Selection
}

\author{Wataru Hashimoto$^{1}$ and Kazumune Hashimoto$^{2}$% <-this % stops a space
\thanks{{Wataru Hashimoto and Kazumune Hashimoto are with the Graduate School of Engineering, The University of Osaka, Suita, Japan (e-mail: hashimoto@is.eei.eng.osaka-u.ac.jp, hashimoto@eei.eng.osaka-u.ac.jp). The corresponding author is Wataru Hashimoto.
This work is supported by JST CREST JPMJCR201, JST ACT-X JPMJAX23CK, and JSPS KAKENHI Grant 21K14184, and 22KK0155.
}}% <-this % stops a space
}

\begin{document}

\maketitle
\thispagestyle{empty}
\pagestyle{empty}

%%%%%%%%%%%%%%%%%%%%%%%%%%%%%%%%%%%%%%%%%%%%%%%%%%%%%%%%%%%%%%%%%%%%%%%%%%%%%%%%
\begin{abstract}
Learning Model Predictive Control (LMPC) improves performance on iterative tasks by leveraging data from previous executions. At each iteration, LMPC constructs a sampled safe set from past trajectories and uses it as a terminal constraint, with a terminal cost given by the corresponding cost-to-go.
While effective, LMPC heavily depends on the initial trajectories: states with high cost-to-go are rarely selected as terminal candidates in later iterations, leaving parts of the state space unexplored and potentially missing better solutions. For example, in a reach-avoid task with two possible routes, LMPC may keep refining the initially shorter path while neglecting the alternative path that could lead to a globally better solution.
To overcome this limitation, we propose Multi-Modal LMPC (MM-LMPC), which clusters past trajectories into modes and maintains mode-specific terminal sets and value functions. A bandit-based meta-controller with a Lower Confidence Bound (LCB) policy balances exploration and exploitation across modes, enabling systematic refinement of all modes. This allows MM-LMPC to escape high-cost local optima and discover globally superior solutions.
We establish recursive feasibility, closed-loop stability, asymptotic convergence to the best mode, and a logarithmic regret bound. Simulations on obstacle-avoidance tasks validate the performance improvements of the proposed method.
\end{abstract}

%%%%%%%%%%%%%%%%%%%%%%%%%%%%%%%%%%%%%%%%%%%%%%%%%%%%%%%%%%%%%%%%%%%%%%%%%%%%%%%%

\section{INTRODUCTION}
Model Predictive Control (MPC) is a widely used control strategy that determines control inputs by repeatedly solving a finite-horizon optimal control problem at each sampling instant based on a predictive model of the system dynamics \cite{MPC}. Its ability to explicitly handle system constraints and multivariate systems has made MPC a powerful tool in various engineering domains, from process control \cite{MPCengineering} to autonomous systems such as robotics and self-driving vehicles \cite{MPCAGV1,MPCLocomotion}.
 However, since MPC determines the control input by solving an optimal control problem over a relatively short prediction horizon, its decisions may deviate from the true infinite-horizon optimal solution. This can lead to high-cost suboptimal performance, particularly in scenarios where long-term effects and delayed consequences play a significant role in achieving the control objectives.

To address this problem, Ugo Rosolia and Francesco Borrelli 
 proposed Learning Model Predictive Control (LMPC) \cite{iterative1}, which repeatedly applies control with MPC to iterative tasks while leveraging state and input trajectories from previous iterations to improve control performance. In their approach, terminal constraints and terminal cost functions are progressively updated using data from successful past iterations, thereby ensuring recursive feasibility of the optimization problem, stability of the closed-loop system, and non-increasing iteration costs under suitable assumptions.
 However, a key limitation of the LMPC framework lies in its strong dependence on the set of trajectories provided in early iterations. Since the terminal constraint and cost are constructed from states visited in previous successful trials, and the terminal constraint in LMPC imposes the system to match with one of the states that has been visited in a previous iteration, the controller can only explore solutions that remain within the regions near these trajectories. For instance, in a navigation task with obstacles, if the initial feasible trajectory passes to the left side of an obstacle, the controller will generally converge to the best path within that left-side corridor even if the globally optimal route lies to the right. Moreover, even when initial trajectories on both sides are provided, if an initial trajectory regarding the right-side path is longer than that of the left-side, it never contributes to the MPC solution due to the high cost-to-go associated with a state in that trajectory.  %This dependence on the initial data can lead to convergence to high-cost suboptimal solutions in problems where multiple qualitatively different solution classes exist.

To overcome this limitation, we propose \emph{Multi-Modal Learning Model Predictive Control (MM-LMPC)}, a framework that systematically explores and exploits multiple solution modes. The approach begins by clustering past trajectories into distinct modes and assigning each mode its own LMPC controller with a dedicated terminal set and value function. A high-level meta-controller, formulated as a multi-armed bandit problem, selects which mode to execute at each iteration. This design balances the refinement of well-performing modes with the exploration of under-explored ones, enabling the controller to escape high-cost local optima and discover globally superior solutions.
Our theoretical analysis shows that MM-LMPC preserves recursive feasibility, ensures closed-loop stability, guarantees asymptotic convergence to the best-performing mode, and achieves a logarithmic regret bound in the number of iterations. Simulation results on a minimum-time reach-avoid problem for the Dubins car demonstrate that the proposed method outperforms the standard LMPC algorithm.

\textbf{Related works on iterative learning MPC:}
The idea of leveraging past execution data to improve control performance in repetitive tasks has long been central to iterative learning control (ILC) \cite{ILC1,ILC2}. More recently, substantial effort has focused on integrating ILC with MPC, enabling explicit state-constraint handling and closed-loop stability guarantees \cite{ILMPC1,ILMPC2,ILMPC3,ILMPC4,ILMPC5,ILMPC6}. An early attempt \cite{ILMPC1} combined ILC with Generalized Predictive Control (GPC), demonstrating significant performance gains, followed by extensions to general nonlinear systems with convergence guarantees to a prescribed reference trajectory \cite{ILMPC2,ILMPC3,ILMPC4,ILMPC5,ILMPC6}. A key limitation of these methods is their reliance on a fixed reference trajectory, limiting practical applicability.

To address this, reference-free iterative learning MPC frameworks have been proposed. These methods iteratively refine the terminal set and cost using trajectory data from previous iterations, approximating the infinite-horizon solution via repeated finite-horizon MPC problems, provided at least one feasible (not necessarily optimal) trajectory is available \cite{iterative1,iterative2}. For linear systems, convergence to the optimal solution is guaranteed, while for nonlinear systems monotonic performance improvement is ensured \cite{iterative1}. This approach has since been generalized to uncertain linear systems \cite{iterative3}, probabilistic nonlinear systems \cite{iterative4}, unknown dynamics \cite{self1,self2,RLMPC}, cooperative multi-agent settings \cite{iterative5}, and certificate-function-based formulations \cite{NNLMPC}, with successful demonstrations in domains such as autonomous racing \cite{AV} and robotic surgery \cite{surgical}.

Particularly relevant is task decomposition MPC (TDMPC) \cite{TDMPC1,TDMPC2}, which leverages the subtask structure of LMPC to build safe sets and terminal costs for new tasks by reordering previously solved subtasks. Our work similarly exploits task structure but focuses on mode decompositions within a single task, combined with a bandit-based mode selection strategy. Prior multi-modal LMPC studies mainly addressed modality due to changes in physical dynamics \cite{MMLMPC}, whereas we target intra-task modal diversity.

\section{Problem Formulation}

We consider a discrete-time nonlinear system
\begin{equation}\label{eq:system}
    x_{t+1} = f(x_t, u_t), \quad x_t \in \mathbb{R}^n, \; u_t \in \mathbb{R}^m,
\end{equation}
subject to state and input constraints
\begin{equation}
    x_t \in \mathcal{X}, \quad u_t \in \mathcal{U}.
\end{equation}

The objective of this paper is to design a feedback control law that solves the infinite-horizon optimal control problem:
\begin{equation}
\begin{aligned}
    \min_{\{u_t\}_{t=0}^\infty} \;& \sum_{t=0}^\infty h(x_t, u_t) \\
    \text{s.t.} \;& x_{t+1} = f(x_t, u_t), \\
                  & x_t \in \mathcal{X}, \quad u_t \in \mathcal{U}, \quad \forall t \ge 0, \\
                  & x_0 \in \mathcal{X}.
\end{aligned}
\end{equation}
where the function $h$ is the stage cost function that encodes the performance of the system.
We make the following assumptions on the system and the stage cost function $h$, which are standard in MPC literature.

\begin{assumption}
\label{assum:system}
The system dynamics $f(\cdot, \cdot)$ are continuous. The state and input constraint sets $\mathcal{X}$ and $\mathcal{U}$ are compact. 
\end{assumption}
\begin{assumption}
The stage cost satisfies $h(x,u) > 0$ for all $x \in \mathcal{X}\backslash \{x_F\}, u\in \mathcal{U}$. where the final state $x_F$ is assumed to be a feasible equilibrium for the unforced system (\ref{eq:system}), i.e., $f(x_F,0)=x_F$. Moreover, the function $h$ satisfies $ h(x_F,0)=0$ and
\begin{align}
    h(x,u)\succ 0, \forall x\in \mathcal{X}\backslash\{x_F\}, u\in \mathcal{U}\backslash \{0\}.  
\end{align}
\end{assumption}

We additionally require the existence of at least one successful feasible trajectory, which is standard in LMPC literature. 
\begin{assumption}[Initial Successful Trajectory]
\label{assum:initial_feasible}
At least one feasible trajectory 
$\{x_0, u_0, x_1, u_1, \dots, x_T\}$ exists that satisfies the system dynamics 
and all state and input constraints, and reaches the final equilibrium 
$x_F$.
\end{assumption}

\section{Review of Learning Model Predictive Control (LMPC)}\label{sec:review}

In this section, we briefly review the LMPC framework \cite{iterative1}, which forms a key foundation of our work. We then introduce a limitation of the original LMPC algorithm that motivates our proposed approach.
In LMPC, the task is executed repeatedly over iterations $j = 0, 1, \dots$ using a finite-horizon MPC. At iteration $j$, a feasible closed-loop trajectory
\[
\{x_0^j, u_0^j, x_1^j, u_1^j, \dots, x_{T_j}^j\}
\]
is obtained, where $T_j$ denotes the time to reach the final state $x_F$.
From all successful previous iterations, LMPC constructs the terminal set as
\begin{equation}\label{eq:SS}
    \mathcal{SS}^j = \bigcup_{i \in M^j} \bigcup_{t=0}^{T_i} x_t^i,
\end{equation}
where $M^j$ is the set of indices of iterations that successfully completed the task before iteration $j$.
For each $x \in \mathcal{SS}^j$, LMPC defines the terminal cost as the minimal cost-to-go among previous visits:
\begin{equation}\label{eq:Q}
    Q^j(x) =
    \begin{cases}
        \displaystyle \min_{(i,t) \in \mathcal{F}^j(x)}
        \; \sum_{k=t}^{T_i} h(x_k^i, u_k^i), & \text{if } x \in \mathcal{SS}^j, \\[2ex]
        +\infty, & \text{otherwise},
    \end{cases}
\end{equation}
where
\begin{equation}
\mathcal{F}^j(x) = \{ (i,t) \mid i \in M^j, ; x_t^i = x \}.
\end{equation}
With the above definitions of terminal set and cost, at time $t$ in iteration $j$, LMPC solves the finite-horizon optimal control problem:
\begin{subequations} \label{eq:mpc}
\begin{align}
    \min_{\{u_{k|t}^j\}_{k=t}^{t+N-1}} \ & \sum_{k=t}^{t+N-1} h(x_{k|t}^j, u_{k|t}^j) + Q^{j-1}(x_{t+N|t}^j) \\
    \text{s.t.}\quad x_{k+1|t}^j &= f(x_{k|t}^j, u_{k|t}^j), \\
    x_{k|t}^j &\in \mathcal{X},\ \ u_{k|t}^j \in \mathcal{U}, \\
    x_{t+N|t}^j &\in \mathcal{SS}^{j-1},\ \ x_{t|t}^j = x_t^j,
\end{align}
\end{subequations}
After solving (\ref{eq:mpc}), the optimal input and corresponding state trajectories are obtained as
\(\{u_{k|t}^{j,*}\}_{k=t}^{t+N-1}\) and \(\{x_{k|t}^{j,*}\}_{k=t}^{t+N}\), respectively.  
Then, the first optimal control input \(u_{t|t}^{j,*}\) is applied to the system (\ref{eq:system}) and the next state \(x_{t+1}^j\) is observed. At the next time step, the optimization is solved again from the initial state \(x_{t+1}^j\). This procedure is repeated at each time step, thereby implementing a receding-horizon control.
After each iteration, the terminal components \(\mathcal{SS}^j\) and \(Q^j\) are updated based on the corrected data according to the definition (\ref{eq:SS}) and (\ref{eq:Q}). As discussed in Section III of \cite{iterative1}, LMPC guarantees desirable properties such as recursive feasibility and stability of the closed-loop system, and ensures that the total cost of each iteration does not increase.

However, the original LMPC algorithm tends to focus exploration on regions associated with previously observed low-cost trajectories, which can be problematic in tasks that admit multiple qualitatively distinct solution modes. For example, consider the reach-avoid problem illustrated in Fig.~\ref{fig:example}, where a vehicle must reach a goal region while avoiding an obstacle. Suppose two feasible initial trajectories are provided, one passing above the obstacle and the other below. Even if the globally optimal solution follows the lower path, LMPC may converge to the suboptimal upper path if the lower trajectory is initially longer, since states along that path may never be selected as terminal states in subsequent iterations. Indeed, in the simulation shown in Fig.~\ref{fig:example}, states along the lower path are never selected as terminal candidates, thus preventing the algorithm from exploring this alternative route. This limitation motivates the development of our proposed approach.

\begin{figure}[t]
    \centering
    \includegraphics[width=\linewidth]{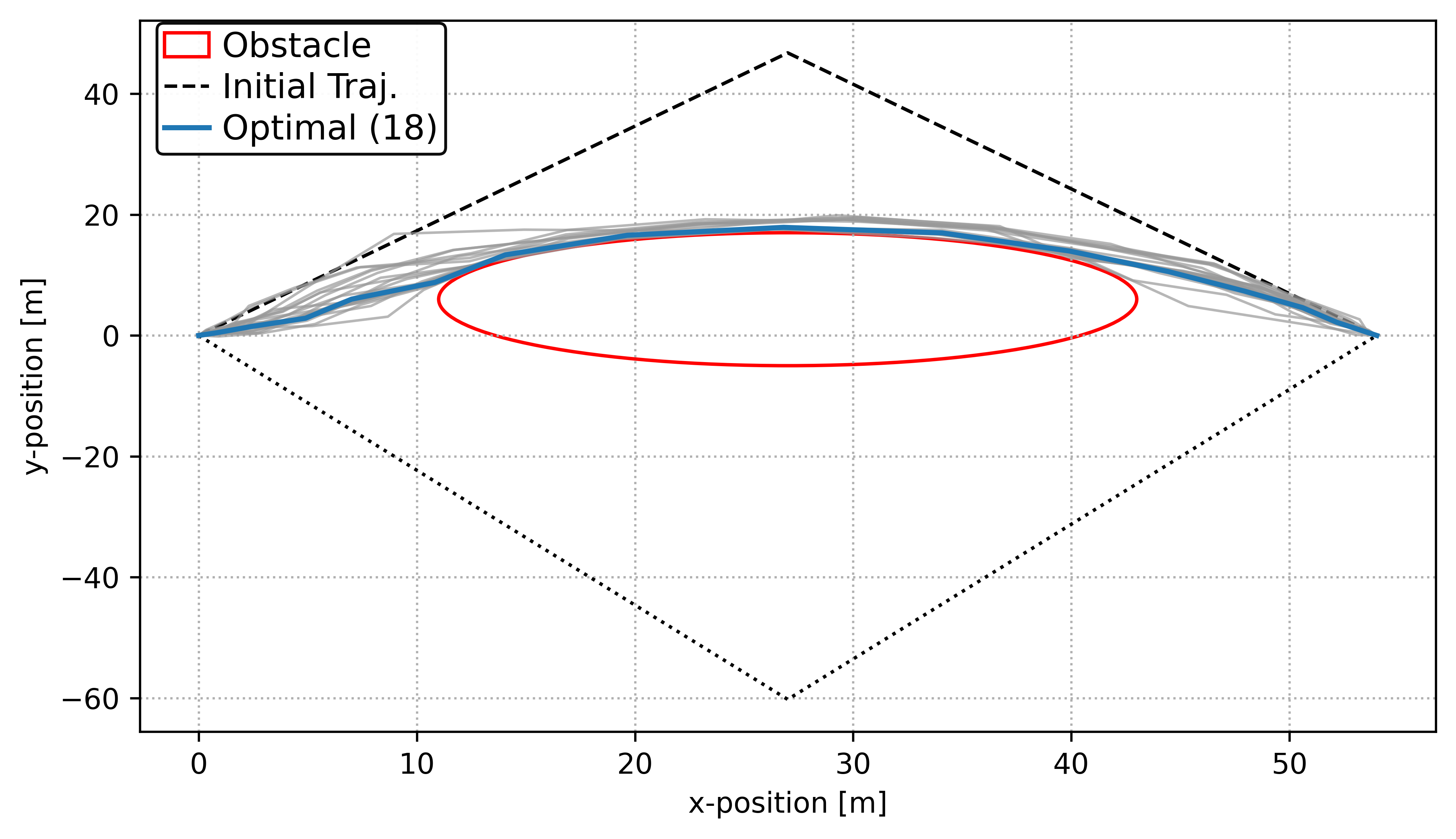}
    \caption{Execution example of standard LMPC: black dashed/dotted are initial seeds, gray are rollouts, bold curve is the final best. Obstacle shown in red.}
    \label{fig:example}
\end{figure}

\section{Proposed Methodology: Multi-Modal Learning Model Predictive Control (MM-LMPC)}
\begin{figure}[t]
    \centering
    \includegraphics[width=\linewidth]{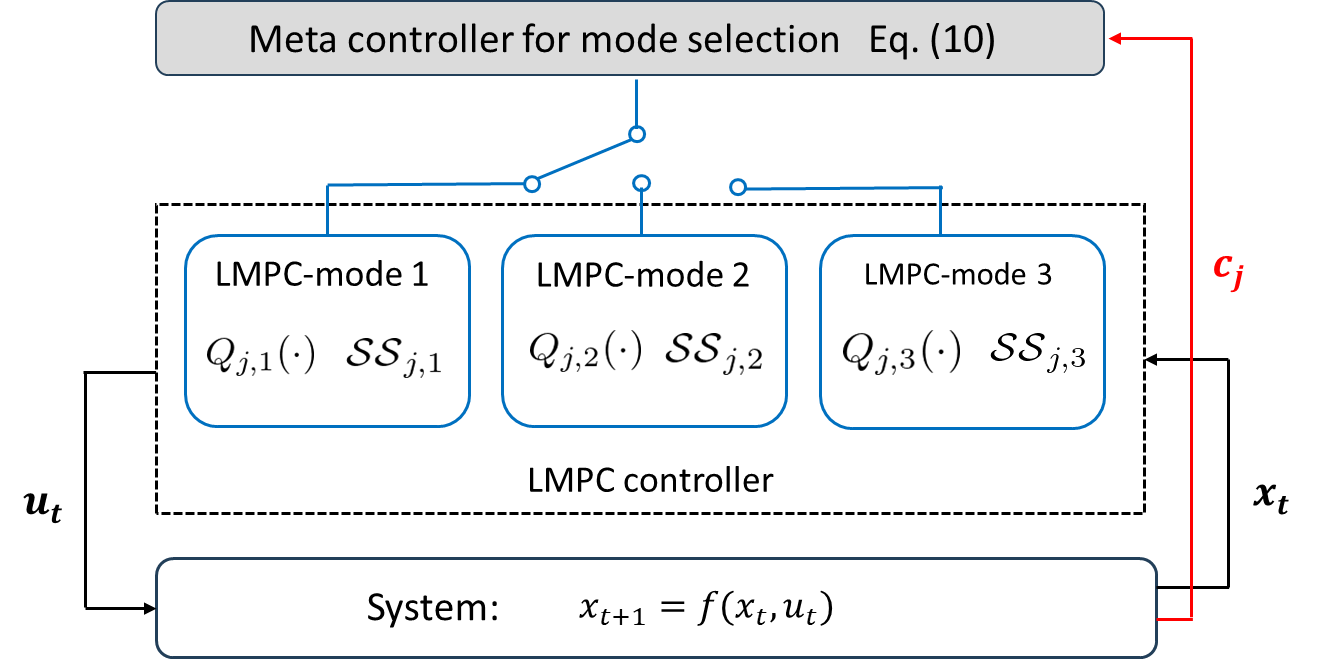}
    \caption{The proposed MM-LMPC architecture.}
    \label{fig:proposed}
\end{figure}

\begin{algorithm}
\DontPrintSemicolon
\caption{Multi-Modal LMPC (MM-LMPC)}
\label{alg:mm-lmpc-final}

\KwIn{$x_0$, initial data $\mathcal{D}_0$, max iterations $J_{max}$, exploration constant $\kappa$}

\textbf{Initialization:}\;
Initialize $n_m\!\leftarrow\!0$, $\mathcal{C}_m^{(0)}\!\leftarrow\!\emptyset$, $\mathcal{SS}_{0,m}\!\leftarrow\!\emptyset$ for all $m$\;
\For{each $(\mathbf{x}^i,\mathbf{u}^i)\in\mathcal{D}_0$}{
    $m_i\!\leftarrow\!\text{Classify}(\mathbf{x}^i)$,\,
    $\mathcal{SS}_{0,m_i}\!\leftarrow\!\mathcal{SS}_{0,m_i}\cup\{x_k^i\}$,\,
    $n_{m_i}\!\leftarrow\!n_{m_i}+1$,\,
    $\mathcal{C}_{m_i}^{(0)}\!\leftarrow\!\mathcal{C}_{m_i}^{(0)}\cup\{J(\mathbf{x}^i,\mathbf{u}^i)\}$\;
}
$N_0 \leftarrow$ number of initialized modes.\;
Construct $Q_{0,m}$ for initialized modes\;

\textbf{Main Loop:}\;
\For{$j=1$ \KwTo $J_{max}$}{
    \textit{1. Mode Selection:}\;
    $j_{total}\!\leftarrow\!\sum_m n_m$,\,
    $\hat{J}_{m}^{(j-1)}\!\leftarrow\!\min(\mathcal{C}_m^{(j-1)})$,\,
    $m_j\!\leftarrow\!\arg\min_m \bigl(\hat{J}_{m}^{(j-1)}-\kappa\sqrt{\log j_{total}/\max\{1,n_m\}}\bigr)$\;
    
    \textit{2. Execute Iteration:}\;
    Generate trajectory $(\mathbf{x}^j,\mathbf{u}^j)$ by solving (\ref{eq:MPC2}) for mode $m_j$\;
    
    \textit{3. Classify and Update:}\;
    $m_{new}\!\leftarrow\!\text{Classify}(\mathbf{x}^j)$,\,
    $N_j\!\leftarrow\!\max(N_{j-1},m_{new})$,\,
    $\mathcal{SS}_{j,m_{new}}\!\leftarrow\!\mathcal{SS}_{j-1,m_{new}}\cup\{x_k^j\}$,\,
    $Q_{j,m_{new}}\!\leftarrow\!\text{Construct from }\mathcal{SS}_{j,m_{new}}$,\,
    $n_{m_{new}}\!\leftarrow\!n_{m_{new}}+1$,\,
    $\mathcal{C}_{m_{new}}^{(j)}\!\leftarrow\!\mathcal{C}_{m_{new}}^{(j-1)}\cup\{J(\mathbf{x}^j,\mathbf{u}^j)\}$\;
    \For{$m\neq m_{new}$}{$\mathcal{SS}_{j,m}\!\leftarrow\!\mathcal{SS}_{j-1,m}$\;}
}
\end{algorithm}

To address the problem of standard LMPC discussed in the previous section, we propose a \emph{Multi-Modal LMPC} (MM-LMPC) architecture that maintains and coordinates multiple LMPC controllers, each specialized for a distinct motion pattern. In the following discussion, we denote the trajectory obtained at each iteration $j$ as $\mathbf{x}^j=\{x_0, x_1, \dots, x_{T^j}\}$ and $\mathbf{u}^j=\{x_0, x_1, \dots, x_{T_j}\}$, respectively, for notational simplicity. Moreover, we denote by $J(\mathbf{x}^j,\mathbf{u}^j)$ the total cost of a closed-loop trajectory of $j$-th iteration. 

The proposed method consists of three components: clustering of the obtained trajectories into modes, control execution with mode-specific LMPC, and a meta-controller for mode selection. 
%Before presenting the overall control framework, we introduce the notation used throughout this section. For each mode $m$, $\mathcal{SS}_{j,m}$ denotes the sampled safe set at iteration $j$, $\mathcal{C}_m$ the set of trajectory costs observed for mode $m$, and $n_m$ the number of times mode $m$ has been executed. We write $J(\mathbf{x},\mathbf{u})$ for the total cost of a closed-loop trajectory $(\mathbf{x},\mathbf{u})$, and $\hat{J}_{m,j}=\min(\mathcal{C}_m)$ for the best cost observed in mode $m$ up to iteration $j$. The total number of iterations is denoted $J_{\max}$, and $j$ indexes the current iteration.
The overall MM-LMPC control architecture and algorithm are illustrated in Fig.~\ref{fig:proposed} and Algorithm~\ref{alg:mm-lmpc-final}. 
These components are described in the following subsections.

\subsection{Mode Clustering}
First, we consider the clustering of the trajectories that have been corrected in previous iterations.
In some applications, the possible solution modes are known in advance (for example, whether a vehicle passes to the left or right of an obstacle).  
In such cases, the modes can simply be specified manually and fixed throughout the learning process, allowing domain knowledge to be directly incorporated.  
When such prior knowledge is not available, MM-LMPC needs to identify modes automatically from historical trajectory data.  
Each stored closed-loop trajectory $(\mathbf{x}^i,\mathbf{u}^i)$ is mapped to a feature vector representation, and an unsupervised clustering method such as DBSCAN or a Gaussian Mixture Model (GMM) \cite{DBSCAN, GMM} can be applied to partition the trajectories into topologically distinct clusters.  
%This automated discovery enables the controller to adaptively recognize new modes as they emerge during iterative executions.
\begin{comment}
Before proceeding, we clarify the initialization requirement for MM-LMPC.  
While any discovered mode by definition comes with at least one generating trajectory, 
we assume that at the beginning of the learning process every mode to be considered is equipped with at least one feasible closed-loop trajectory. 
This ensures that mode-specific sampled safe sets and value functions are well-defined from the start.
\begin{assumption}[Given Initial Feasible Trajectories]
\label{assum:initial_trajectory}
For each mode $m$ present at initialization, at least one complete feasible closed-loop trajectory $(\mathbf{x}^i,\mathbf{u}^i)$ is provided, 
satisfying the system dynamics and all input and state constraints.
\end{assumption}
\end{comment}

\subsection{Mode-Specific LMPC}
For each mode $m \in \{1,\dots,N_j\}$, MM-LMPC instantiates a dedicated LMPC controller.  
Using only the trajectory data associated with mode $m$, we construct a \emph{mode-specific sampled safe set} $\mathcal{SS}_{j,m}$ and a \emph{mode-specific value function} $Q_{j,m}(\cdot)$. The definitions of $\mathcal{SS}_{j,m}$ and $Q_{j,m}(\cdot)$ follow (\ref{eq:SS}) and (\ref{eq:Q}), respectively.
At time $t$ of iteration $j$, the controller corresponding to the selected mode $m$ solves the following finite-horizon optimal control problem:
\begin{align}\label{eq:MPC2}
\min_{\{u_{k|t}\}} \quad & \sum_{k=t}^{t+N-1} h(x_{k|t}, u_{k|t}) + Q_{j-1,m}(x_{t+N|t}) \nonumber \\
\text{s.t.} \quad & x_{k+1|t} = f(x_{k|t}, u_{k|t}), \nonumber \\
& x_{k|t} \in \mathcal{X}, \quad u_{k|t} \in \mathcal{U}, \nonumber \\
& x_{t+N|t} \in \mathcal{SS}_{j-1,m}, \nonumber \\
& x_{t|t} = x_t^j.
\end{align}
The resulting control inputs are applied in the same receding-horizon manner as in the standard LMPC described in Section~\ref{sec:review}, with the mode $m$ fixed throughout the iteration. In the following discussion, we denote by $c_{m,k}$ the total cost of the closed-loop trajectory obtained when mode $m$ is executed for the $k$-th time.
 
%Each controller thus optimizes within its homotopy class, while preserving the recursive feasibility of LMPC.
\subsection{Meta-Controller for Mode Selection via Multi-Armed Bandits}

For each mode $m$, let $\mathcal{C}_m^{(j)} := \{c_{m,k} \mid k \le j,\, m_k = m\}$ denote 
the set of trajectory costs observed for mode $m$ up to iteration $j$,
and define $\hat{J}_{m}^{(j)} := \min(\mathcal{C}_m^{(j)})$ as the best cost observed for mode $m$ up to iteration $j$. 
Let $n_m(j) := |\mathcal{C}_m^{(j)}|$ be the number of executions of mode $m$ up to iteration $j$. 
The choice of which mode to execute at the beginning of iteration $j+1$ is posed as a 
\emph{multi-armed bandit} (MAB) problem \cite{bandit1, bandit2}, with each mode treated as an arm. 
Specifically, we adopt a \emph{Lower Confidence Bound} (LCB) rule \cite{UCB}:
\begin{equation}
    m^*_{j+1} = \arg\min_{m \in \{1,\dots,N_j\}}
    \Bigl(
        \hat{J}_{m}^{(j)} - 
        \kappa \sqrt{\frac{\log(j+1)}{\max\{1,n_m(j)\}}}
    \Bigr),
\end{equation}
where $\kappa > 0$ controls the exploration–exploitation trade-off. 
The first term $\hat{J}_{m}^{(j)}$ directs the controller toward modes that have demonstrated low costs in the past, capturing the exploitation aspect of the policy, while the second term provides an exploration bonus that prioritizes less-tested modes. 
Together, these terms allow MM-LMPC not only to refine well-performing modes but also to repeatedly explore alternative ones, thereby ensuring the discovery of globally competitive solutions in the long run.

\subsection{Summary of the Proposed Algorithm}
The proposed MM-LMPC algorithm is summarized in Algorithm~\ref{alg:mm-lmpc-final}. 
By Assumption~\ref{assum:initial_feasible}, the learning process starts with at least one successful trajectory, 
ensuring that the initialization phase of the algorithm is well-defined.

%Then, we move on to the detailed explanation of the procedures in the proposed method.
In the initialization phase (lines 2–5), the algorithm initializes each mode with an empty cost set $\mathcal{C}_m^{(0)}$, 
an empty sampled safe set $\mathcal{SS}_{0,m}$, and a counter $n_m(0)=0$.  
For every initial trajectory $(\mathbf{x}^i,\mathbf{u}^i) \in \mathcal{D}_0$, the trajectory is classified into a mode $m_i$ (line 4), 
after which the corresponding sets and counters are updated: the visited states are added to $\mathcal{SS}_{0,m_i}$, 
the counter $n_{m_i}(0)$ is incremented, and the trajectory cost is inserted into $\mathcal{C}_{m_i}^{(0)}$ (line 4).  
After all initial trajectories have been processed, the total number of initialized modes $N_0$ is determined (line 5) 
and the initial terminal cost for each mode $Q_{0,m}$ is defined based on $\mathcal{SS}_{0,m_i}$.

In the main loop (lines 8–16), repeated for $j = 1,\dots,J_{\max}$, three steps are executed.  
First, a mode $m_j$ is selected using the LCB rule (line 10), 
where the best observed cost for each mode is computed as $\hat{J}_m^{(j-1)} = \min(\mathcal{C}_m^{(j-1)})$.  
Second, the LMPC for the selected mode $m_j$ is executed and generates a closed-loop trajectory $(\mathbf{x}^j,\mathbf{u}^j)$ from $x_0$ (line 12).  
Third, the new trajectory is classified into a mode $m_{\mathrm{new}}$, 
its safe set $\mathcal{SS}_{j,m_{\mathrm{new}}}$ and terminal cost $Q_{j,m_{\mathrm{new}}}$ are updated, 
the counter $n_{m_{\mathrm{new}}}(j)$ is increased, 
and the observed cost is added to $\mathcal{C}_{m_{\mathrm{new}}}^{(j)}$ (line 14), 
while the safe sets of all other modes are carried over unchanged (lines 15–16).  

%By iterating this process, MM-LMPC progressively refines the controllers specialized for each mode 
%while ensuring systematic exploration of less-tested ones, thereby overcoming the limitations of standard LMPC.

\section{Theoretical Analysis}

In this section, we provide a theoretical analysis of the proposed MM-LMPC framework. 
We establish guarantees for recursive feasibility, stability, and convergence, and further analyze the regret associated with the bandit-based mode selection. 
Our analysis builds upon the foundational properties of LMPC \cite{iterative1} and extends them to our multi-modal, bandit-driven architecture. 
The analysis relies on the following assumptions.

\begin{assumption}[Finiteness of Modes]
\label{assum:finiteness}
Let $M_j$ denote the set of modes discovered up to iteration $j$. We assume that this set converges to a finite set $M_\infty$ as the number of iterations $j$ tends to infinity.
\end{assumption}

\begin{assumption}[Classifier Consistency]
\label{assum:classifier_consistency}
After a sufficient number of iterations, the trajectory classification becomes consistent. 
That is, for each mode $m$, there exists an iteration $J_c$ such that for all $j > J_c$, any new trajectory generated by executing mode $m$ will consistently be classified into mode $m$. 
%This ensures that the data within each mode’s safe set remains coherent.
\end{assumption}

\begin{assumption}[Intra-Mode Convergence Rate]
\label{assum:conv_rate}
For any mode $m \in M_\infty$, let $c_{m,k}$ be the cost of the trajectory generated the $k$-th time that mode $m$ is executed. 
We assume that the cost converges to its optimal value $c_m^*$, such that the sequence of cost improvements $\delta_{m,k} = c_{m,k} - c_m^*$ is summable. 
That is,
\begin{equation}
    \sum_{k=1}^{\infty} (c_{m,k} - c_m^*) \le C_m < \infty,
\end{equation}
where $C_m$ is a finite constant depending on the mode.
\end{assumption}

\noindent
Assumption~\ref{assum:finiteness} is natural in planning and control problems where the number of qualitatively distinct solution patterns is finite. 
Assumption~\ref{assum:classifier_consistency} is reasonable in practice, since many clustering and feature extraction methods exhibit stable behavior once sufficient data has been accumulated. 
This stability is crucial for ensuring that each mode’s safe set and value function are updated coherently.
Finally, Assumption~\ref{assum:conv_rate} strengthens the standard LMPC property that costs are non-increasing (see Lemma \ref{lem:intra_mode_nonincreasing} later) and bounded below by the nonnegative stage cost. 
The additional requirement that the improvement sequence be summable is not restrictive in practice, since it simply rules out pathological cases of arbitrarily slow convergence and ensures that the cumulative deviation from the optimal cost remains finite.

With these assumptions, we can establish the main theoretical properties of the MM-LMPC framework.

\begin{theo}[Recursive Feasibility and Stability]
\label{thm:feasibility}
Under Assumptions \ref{assum:system}--\ref{assum:classifier_consistency}, 
the MM-LMPC controller is recursively feasible for all iterations $j \geq 1$ and time steps $t \geq 0$. 
Furthermore, for each fixed iteration $j$, the closed-loop system is asymptotically stable. 
\end{theo}

\begin{proof}
At the beginning of iteration $j$, the meta-controller selects a mode $m_j \in M_{j-1}$. 
For that iteration, the controller operates exactly as a standard LMPC with the corresponding sampled safe set $\mathcal{SS}_{j-1,m_j}$ and value function $Q_{j-1,m_j}$.  
Recursive feasibility can be established following the standard LMPC argument. 
At $t=0$, feasibility is guaranteed because $\mathcal{SS}_{j-1,m_j}$ contains at least one complete closed-loop trajectory from a previous execution, which can be used directly as a candidate solution of (\ref{eq:MPC2}). 
For $t>0$, let the optimal input sequence at time $t-1$ be 
$\{u_{k|t-1}^{*}\}_{k=t-1}^{t+N-2}$ with corresponding state sequence 
$\{x_{k|t-1}^{*}\}_{k=t-1}^{t+N-1}$. 
At time $t$, we can construct a feasible candidate by taking 
\[
\tilde u_{k|t} = u_{k|t-1}^{*}, \quad k=t,\dots,t+N-2,
\]
together with the terminal input $\tilde u_{t+N-1|t}$ that drives 
$x_{t+N-1|t-1}^{*}$ into some $x_{t+N|t}\in\mathcal{SS}_{j-1,m_j}$. 
The corresponding state sequence 
$\{\tilde x_{k|t}\}_{k=t}^{t+N}$ is feasible for (\ref{eq:MPC2}), 
since $\tilde x_{t+N|t}\in\mathcal{SS}_{j-1,m_j}$ by construction. 
Hence feasibility is preserved for all $t\ge0$.
  
Asymptotic stability of the closed-loop system also follows directly from Theorem~1 of \cite{iterative1}. 
\end{proof}

The next lemma establishes a key property of LMPC that will be repeatedly used in our analysis.
Under Assumption 5 (classifier consistency), once trajectory classification has stabilized (i.e., for all iterations $j > J_c$), the realized cost within any fixed mode does not increase across successive executions of that mode.
\begin{comment}
\begin{lemma}[Intra-Mode Non-Increasing Cost]
\label{lem:intra_mode_nonincreasing}
For any fixed mode $m \in M_\infty$, let $c_{m,k}$ denote the cost of the $k$-th trajectory generated by executing mode $m$. 
Then the sequence $\{c_{m,k}\}_{k=1}^\infty$ is non-increasing:
\begin{equation}
    c_{m,k+1} \leq c_{m,k}, \quad \forall k \geq 1.
\end{equation}
\end{lemma}

\begin{proof}
The claim of this lemme also follows directly from the standard LMPC result (Theorem~2 in \cite{iterative1}), 
which establishes that the realized cost does not increase across iterations. 
Applying this argument within each fixed mode yields the desired property.
\end{proof}
\end{comment}
\begin{lemma}[Intra-Mode Non-Increasing Cost]\label{lem:intra_mode_nonincreasing}
Suppose Assumptions~1–5 hold, and let $J_c$ be the iteration index guaranteed by Assumption~5 (Classifier Consistency). 
Fix any mode $m \in \mathcal{M}_\infty$ and consider the sequence of closed-loop iteration costs $\{c_{m,k}\}_{k\ge1}$ obtained by executing mode $m$ for the $k$-th time at iterations strictly after $J_c$. 
Then the sequence is non-increasing:
\begin{align}
c_{m,k+1} \le c_{m,k}, \qquad \forall k \ge 1.
\end{align}
\end{lemma}
\begin{proof}
For $j>J_c$, Assumption~5 ensures that any trajectory generated while executing mode $m$ is consistently classified into the same mode $m$, so the mode-specific sampled safe set $SS_{j,m}$ and terminal cost $Q_{j,m}$ are updated coherently. 
Hence, the standard LMPC monotonicity argument (applied per mode) carries over verbatim: using the shifted optimal solution at time $t-1$ as a feasible candidate at time $t$ shows that the realized iteration cost within mode $m$ cannot increase from one execution to the next (see Theorem~2 in \cite{iterative1}). 
Therefore $c_{m,k+1}\le c_{m,k}$ for all $k\ge1$.
\end{proof}

Then, the following theorem establishes the asymptotic optimality of the proposed method.
\begin{theo}[Asymptotic Optimality]
\label{thm:optimality}
%Let the algorithm converge. Under Assumptions \ref{assum:system}-\ref{assum:classifier_consistency}, the cost of the converged trajectory, $J_{0 \to \infty}^\infty(x_S)$, is equal to the minimum of the optimal costs achievable within each of the discovered modes:
Under Assumptions \ref{assum:system}--\ref{assum:conv_rate} and \ref{assum:classifier_consistency}, 
the closed-loop cost of the converged trajectory $J_{0 \to \infty}^\infty(x_S)$ equals the minimum 
mode-wise optimal cost:
\begin{equation}
J_{0 \to \infty}^\infty(x_S) = \min_{m \in M_\infty} \left( J_{0 \to \infty, m}^* \right)
\end{equation}
where $J_{0 \to \infty, m}^*$ is the optimal cost achievable within the mode $m$.
\end{theo}
\begin{proof}
The proof proceeds by first establishing convergence within each mode. 
For any mode $m \in M_\infty$ that is selected infinitely often, the realized cost sequence 
(after trajectory classification has stabilized) is non-increasing and therefore converges to a well-defined limit. 
We denote this limit by $J_{0 \to \infty, m}^*$. 
This follows directly from Lemma~\ref{lem:intra_mode_nonincreasing}, which guarantees monotonicity of the iteration costs, 
together with the fact that costs are nonnegative and thus bounded below.

Next, we show that all discovered modes must indeed be selected infinitely often. 
Assumption~\ref{assum:finiteness} ensures that the set of modes $M_\infty$ is finite. 
Suppose, for the sake of contradiction, that some mode $m$ were selected only finitely many times. 
Then its counter $n_{m,j}$ would eventually remain constant, while the exploration bonus in the LCB policy 
continues to grow without bound as $j \to \infty$, eventually making mode $m$'s LCB strictly smaller 
than that of any other mode and forcing its reselection. 
This contradiction implies that every mode must be chosen infinitely often.

With these properties established, the remaining argument is straightforward. 
Since every mode is explored infinitely often, the exploration bonus in the LCB policy 
vanishes for all modes as $j \to \infty$, so the selection policy becomes asymptotically greedy 
and chooses the mode with the smallest empirically estimated cost $\hat{J}_{m,j}$. 
Because $\hat{J}_{m,j}$ converges to the true limit cost $J_{0 \to \infty, m}^*$ for each mode, the algorithm eventually selects the mode with the minimum cost among all discovered modes, 
which proves the theorem.

\end{proof}

While the preceding theorem guarantees asymptotic convergence, the following results analyze the finite-time behavior of the algorithm by characterizing both its single-step performance during exploration and its cumulative performance loss (regret).

\begin{theo}[Bound on Iteration Cost under LCB Selection]
\label{thm:cost_increase_bound}
Let $c_{best}^{(j-1)} = \min_{m \in M_{j-1}} \min(\mathcal{C}_m^{(j-1)})$ be the minimum iteration cost observed across all modes prior to iteration $j$. 
Suppose the LCB policy at iteration $j$ selects mode $m_j$, and let $c_j$ be the resulting iteration cost. Then
\begin{equation}
    c_j \;\le\; c_{best}^{(j-1)} 
    + \kappa \!\left( \sqrt{\tfrac{\log j}{\,n_{m_j}(j-1)\,}} 
    - \sqrt{\tfrac{\log j}{\,n_{m_{best}}(j-1)\,}} \right),
\end{equation}
where $m_{best} \in \arg\min_{m \in M_{j-1}} \min(\mathcal{C}_m^{(j-1)})$ and $n_{m}(j-1)$ is the number of times mode $m$ has been selected prior to iteration $j$.
\end{theo}

\begin{proof}
Define the best observed cost for mode $m$ up to iteration $j-1$ by 
$\hat c_{m}^{(j-1)} := \min(\mathcal{C}_m^{(j-1)})$. 
By the intra-mode non-increasing property (Lemma~\ref{lem:intra_mode_nonincreasing}), the realized cost at iteration $j$ satisfies
\begin{equation}
    c_j \;\le\; \hat c_{m_j}^{(j-1)}. 
    \label{eq:bound_step1}
\end{equation}
Since $m_j$ is chosen by the LCB rule, its LCB score is no larger than that of the best-known mode $m_{best}$:
\begin{equation}
    \hat c_{m_j}^{(j-1)} - \kappa \sqrt{\tfrac{\log j}{\,n_{m_j}(j-1)\,}}
    \;\le\; 
    \hat c_{m_{best}}^{(j-1)} - \kappa \sqrt{\tfrac{\log j}{\,n_{m_{best}}(j-1)\,}}.
\end{equation}
Rearranging and using $\hat c_{m_{best}}^{(j-1)} = c_{best}^{(j-1)}$ yields
\begin{equation}
    \hat c_{m_j}^{(j-1)} \;\le\; c_{best}^{(j-1)} 
    + \kappa \!\left( \sqrt{\tfrac{\log j}{\,n_{m_j}(j-1)\,}} 
    - \sqrt{\tfrac{\log j}{\,n_{m_{best}}(j-1)\,}} \right).
    \label{eq:bound_step2}
\end{equation}
Combining \eqref{eq:bound_step1} and \eqref{eq:bound_step2} gives the stated inequality.
\end{proof}

\begin{theo}[Logarithmic Regret Bound]
\label{thm:regret}
Under Assumptions \ref{assum:system}--\ref{assum:conv_rate}, the cumulative regret $R_T$ of the MM-LMPC algorithm after $T$ iterations, defined as
\begin{align}
    R_T &= \sum_{j=1}^{T}(c_{j}-c^*),
\end{align}
satisfies the following bound:
\begin{align}
    R_T \;&\le\; \sum_{m:\, \Delta_m > 0} \left( \frac{4\kappa^2}{\Delta_m} \log T + C_0 \Delta_m \right) + \sum_{m \in M_{\infty}} C_m \\
    &= O(\log T). \notag
\end{align}
where $c_j$ is the realized cost at iteration $j$, $c^*=\min_{m\in M_\infty} c_m^*$ is the true optimal cost, $\Delta_m = c_m^* - c^* > 0$ is the suboptimality gap, and $C_0, C_m$ are constants independent of $T$.
\end{theo}

\begin{proof}
The proof proceeds by decomposing the cumulative regret $R_T$ into two components: (A) the regret from suboptimal mode selection, and (B) the intra-mode cost gap, representing the temporary suboptimality incurred before convergence within each mode.
\begin{align}
R_T
&= \underbrace{\sum_{j=1}^{T} (c_{m_j}^* - c^*)}_{\text{(A) Suboptimal Selection Regret}}
+ \underbrace{\sum_{j=1}^{T} (c_{j} - c_{m_j}^*)}_{\text{(B) Intra-Mode Cost Gap}}. \label{eq:regret_decomposition}
\end{align}

For the intra-mode term (B), Assumption~\ref{assum:conv_rate} yields
\begin{align}
\sum_{j=1}^{T} (c_{j} - c_{m_j}^*)
&= \sum_{m \in M_\infty}\sum_{k=1}^{n_m(T)} (c_{m,k}-c_m^*) \notag \\
&\le \sum_{m \in M_\infty} C_m,
\end{align}
which is a $T$-independent finite constant.
For the suboptimal selection term (A), let $m^*$ be the optimal mode and fix any suboptimal mode $m$ with gap $\Delta_m:=c_m^*-c^*>0$. At iteration $j$, mode $m$ can be selected only if its empirical best cost, adjusted by the exploration bonus, is no larger than that of $m^*$. Since the empirical best cost of $m$ is at least $c_m^*$ by Lemma~\ref{lem:intra_mode_nonincreasing}, this condition implies
\begin{align}
c_m^*-\kappa \sqrt{\tfrac{\log j}{\,n_m(j-1)\,}}
\;\le\; \hat{c}_{m^*}^{j-1}-\kappa \sqrt{\tfrac{\log j}{\,n_{m^*}(j-1)\,}}. \label{eq:nec_raw}
\end{align}
where as defined in Therem \ref{thm:cost_increase_bound}, $\hat c_{m}^{(j-1)} := \min(\mathcal{C}_m^{(j-1)})$.
By the infinite-selection property of each mode (cf. the argument following Assumption~\ref{assum:finiteness}), the optimal mode $m^*$ is sampled infinitely often. Let $\{c_{m^*,k}\}_{k\ge1}$ denote the sequence of iteration costs when $m^*$ is executed for the $k$-th time. By Lemma~\ref{lem:intra_mode_nonincreasing} the sequence is non-increasing and, by Assumption~\ref{assum:conv_rate}, converges to $c^*$. Hence, for any $\varepsilon\in(0,\Delta_m)$ there exists $K_\varepsilon$ such that $c_{m^*,k}\le c^*+\varepsilon$ for all $k\ge K_\varepsilon$. Since $m^*$ is selected infinitely often, there exists $J_\varepsilon$ with $n_{m^*}(j-1)\ge K_\varepsilon$ for all $j\ge J_\varepsilon$, and therefore the empirical best cost satisfies $\hat c_{m^*}^{(j-1)}\le c^*+\varepsilon$ for all $j\ge J_\varepsilon$.
 Dropping the nonpositive exploration term of $m^*$ on the right of \eqref{eq:nec_raw} then gives
\begin{align}
\Delta_m-\varepsilon \;\le\; \kappa \sqrt{\tfrac{\log j}{\,n_m(j-1)\,}} \qquad (j\ge J_\varepsilon).
\label{eq:necessary_condition_eps}
\end{align}
Thus, a suboptimal mode $m$ can only be selected at sufficiently large $j$ if \eqref{eq:necessary_condition_eps} holds.

To convert \eqref{eq:necessary_condition_eps} into a counting bound, let $\tau_s$ denote the iteration index at which mode $m$ is selected for the $s$-th time; then $n_m(\tau_s-1)=s-1$. Applying \eqref{eq:necessary_condition_eps} at $j=\tau_s$ (for $\tau_s\ge J_\varepsilon$) gives
\begin{align}
s-1 \le \frac{\kappa^2}{(\Delta_m-\varepsilon)^2}\log \tau_s.
\end{align}
Therefore, for any horizon $T$, each $s$ with $\tau_s\le T$ satisfies
\begin{align}
s \le \frac{\kappa^2}{(\Delta_m-\varepsilon)^2}\log T + 1.
\end{align}
By definition, $n_m(T)=\max\{s:\tau_s\le T\}$, hence
\begin{align}
n_m(T)\;\le\;\frac{\kappa^2}{(\Delta_m-\varepsilon)^2}\,\log T + C_0. \label{eq:pull_bound_final}
\end{align}
This bound on the number of pulls holds for any $\varepsilon \in (0, \Delta_m)$. To obtain a concrete and tight bound, we can strategically choose a value for $\varepsilon$. A standard choice that balances the terms in the denominator is to set $\varepsilon = \Delta_m / 2$. Substituting this into \eqref{eq:pull_bound_final}, the denominator becomes $(\Delta_m - \Delta_m/2)^2 = (\Delta_m/2)^2 = \Delta_m^2/4$. This yields a simplified bound for $n_m(T)$:
\begin{align}
    n_m(T) \;\le\; \frac{\kappa^2}{\Delta_m^2/4}\,\log T + C_0 \;=\; \frac{4\kappa^2}{\Delta_m^2}\,\log T + C_0. \label{eq:pull_bound_clean}
\end{align}
Now, substituting this bound into the expression for term (A) yields:
\begin{align}
\sum_{m:\,c_m^*>c^*}& n_m(T)\,\Delta_m
\;\le\; \sum_{m:\,c_m^*>c^*}\left(\frac{4\kappa^2}{\Delta_m^2}\,\log T + C_0\right)\Delta_m \notag \\
&= \left( \sum_{m:\,c_m^*>c^*} \frac{4\kappa^2}{\Delta_m} \right) \log T + \sum_{m:\,c_m^*>c^*} C_0 \Delta_m \notag \\
&= O(\log T).
\end{align}
Combining the constant bound for (B) with this logarithmic bound for (A) gives the final result $R_T=O(\log T)$. This completes the proof.
\end{proof}

\section{Simulation Study}
To evaluate the effectiveness of the proposed method, we conduct a numerical experiment designed to highlight a key limitation of standard LMPC, and demonstrate how the proposed method overcomes it, which was also discussed in Section \ref{sec:review}. The simulation is implemented using the publicly available LMPC repository \cite{LMPCcode}, and the nonlinear MPC problems are solved numerically using CasADi \cite{casadi}.

\subsection{Experimental Setup}
We consider the minimum-time reach-avoid problem of the Dubins car with bounded acceleration, same as the original LMPC paper \cite{iterative1}:
\begin{subequations}\label{eq:Dubins}
    \begin{align}
    J_{0\rightarrow \infty}^*&(x_S)=\min_{\begin{smallmatrix} \theta_0, \theta_1,\ldots\\ a_0,a_1,\ldots \end{smallmatrix}} \sum\limits_{k=0}^{\infty} \mathds{1}_k \label{eq:Dubins0}\\
    \textrm{s.t.} \quad
    &x_{k+1} = \begin{bmatrix} z_{k+1} \\ y_{k+1} \\ v_{k+1} \end{bmatrix}= \begin{bmatrix} z_{k} \\ y_{k} \\ v_{k} \end{bmatrix} + \begin{bmatrix} v_k \cos(\theta_k)\\ v_k \sin(\theta_k)\\ a_k \end{bmatrix},\label{eq:Dubins1}\\
    &x_0=x_S = [0~0~0]^T,\label{eq:Dubins2}\\
    &-s \leq a_k \leq s, ~~\forall k\geq 0 \label{eq:Dubins3}\\
    & \frac{(z_k - z_{obs})^2}{a_e^2} + \frac{(y_k - y_{obs})^2}{b_e^2} \geq 1, ~~ \forall k\geq 0.\label{eq:Dubins5}
    \end{align}
\end{subequations}
Here, the stage cost $h(x_k, u_k)$ in \eqref{eq:Dubins0} is given by the indicator function $\mathds{1}_k$, which is defined as
\begin{equation}\label{eq:Indicator}
\mathds{1}_k = 
\begin{cases} 
1, & \text{if } x_k \neq x_F, \\
0, & \text{if } x_k = x_F,
\end{cases}
\end{equation}
where $x_F = [54, 0, 0]^T$ is the target state. In \eqref{eq:Dubins3}, the acceleration bound is set to $s=1$. The state vector $x_k = [z_k, y_k, v_k]^T$ contains the vehicle's position and velocity, while the control inputs are the heading angle $\theta_k$ and the acceleration $a_k$. 
An elliptical obstacle is placed at $(z_{obs}, y_{obs}) = (27, 6)$ with axes $a_e=16$ and $b_e=11$, 
creating two feasible paths: one passing above and one below the obstacle. 
We generate one initial trajectory for each path (costs 45 and 50, respectively) via brute-force search. 
Although the above path is initially shorter, the globally optimal solution is the below path (see Fig.~\ref{fig:example} or Fig.~\ref{acc_mm_exploration}).
For this setting, we execute the control with both the original LMPC algorithm \cite{iterative1} and the proposed method. The iteration number is set to 20 for both cases.

\subsection{Results}

\begin{figure}[t]
    \centering
    \includegraphics[width=0.95\linewidth]{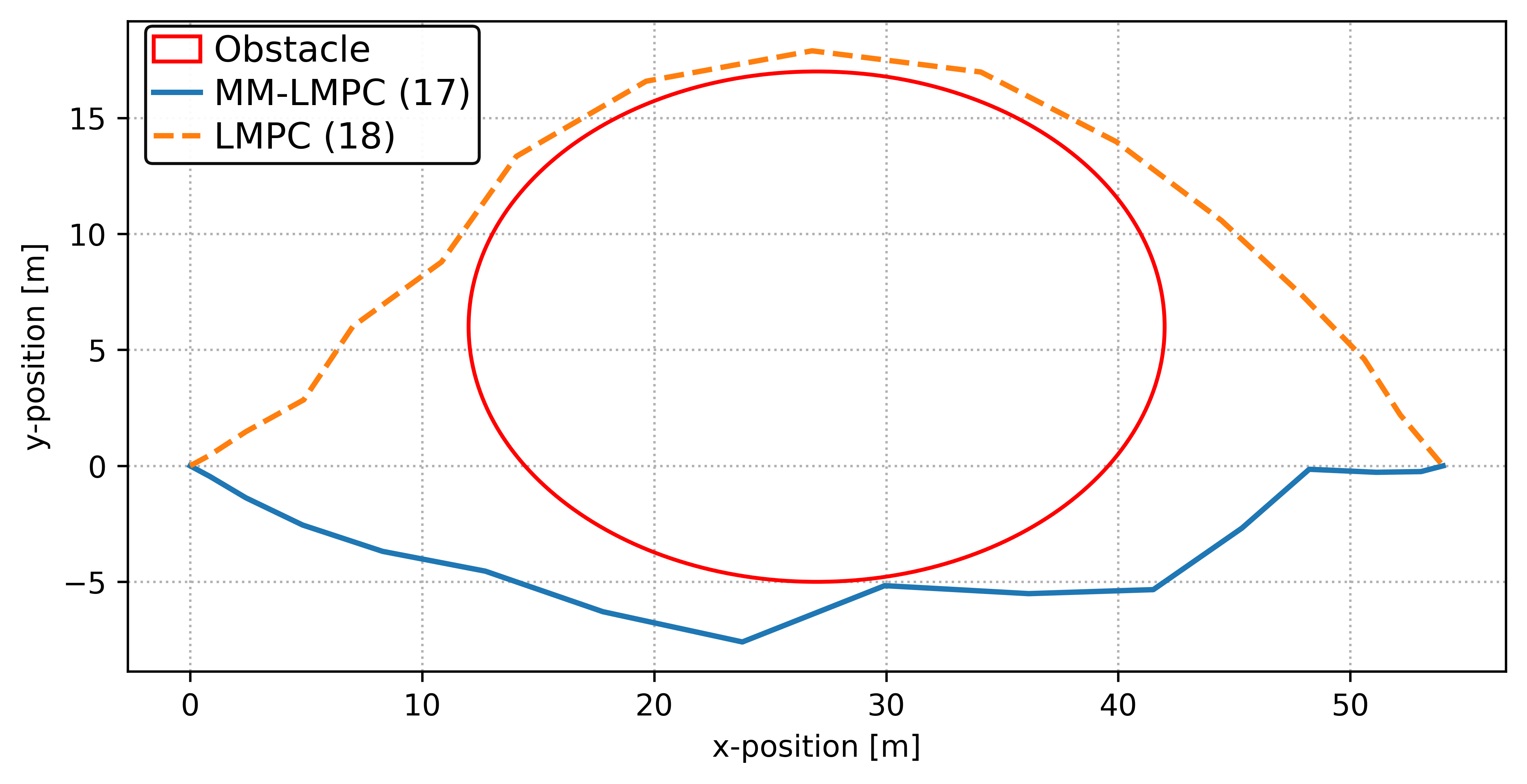}
    \caption{The trajectories obtained at the last iteration (blue: MM-LMPC, orange dashed: LMPC). Red ellipse: obstacle.}
    \label{acc_final_trajectories}
\end{figure}

\begin{figure}[t]
    \centering
    \includegraphics[width=0.95\linewidth]{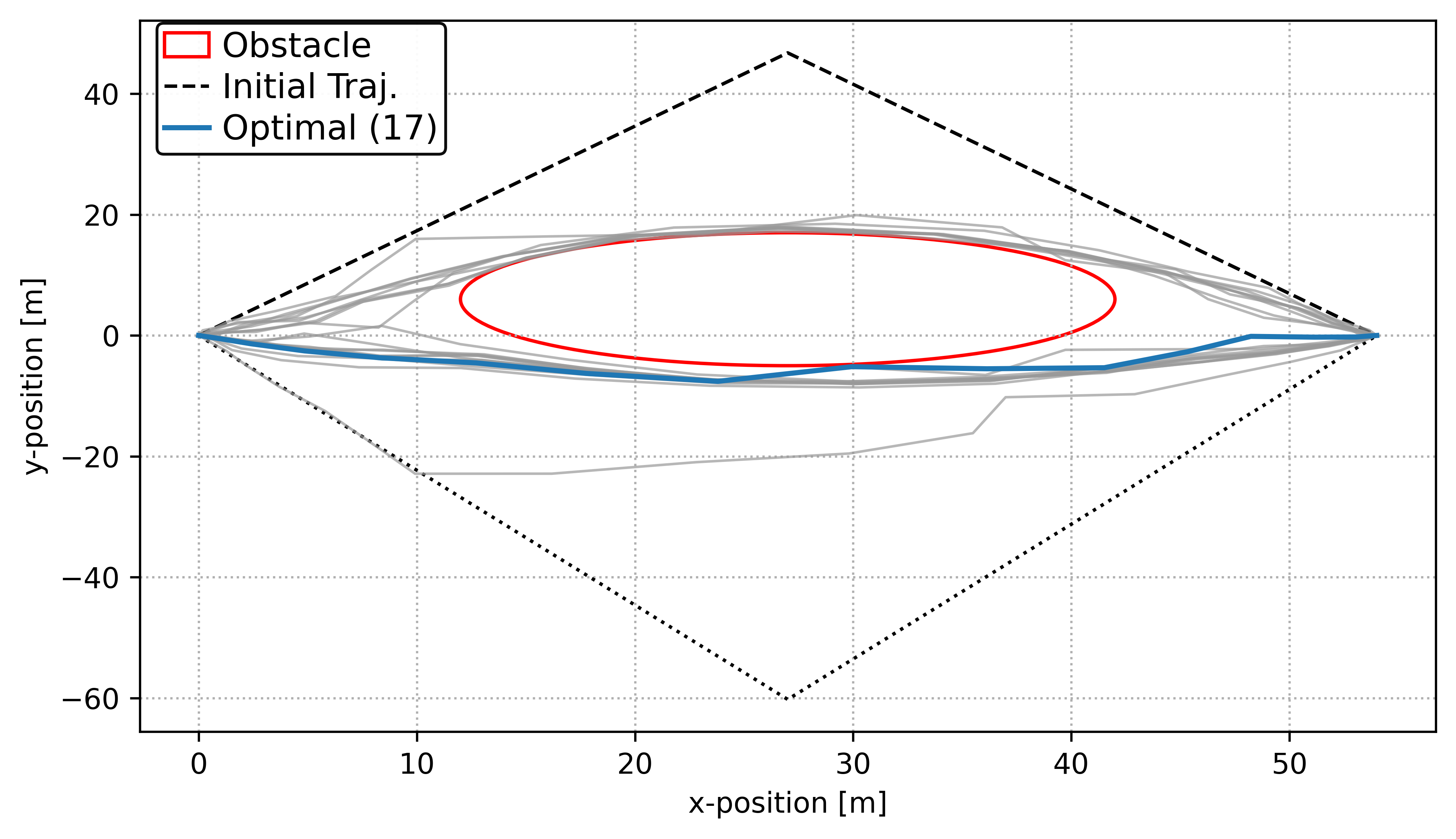}
    \caption{MM-LMPC trajectories during learning (gray), initial trajectories (black dashed/dotted), and the final best path (blue). Red ellipse: obstacle.}
    \label{acc_mm_exploration}
\end{figure}
\begin{comment}
\begin{figure}[t]
    \centering
    \includegraphics[width=\linewidth]{acc_cost_curve.png}
    \caption{Performance cost per iteration. Blue: MM-LMPC; orange: LMPC. Iteration 0 shows the shared initial trajectory cost.}
    \label{acc_final_trajectories}
\end{figure}
\end{comment}

Standard LMPC, initialized with both trajectories, builds a single safe set $\mathcal{SS}^j$ by pooling states from all past data. 
Because the below-path trajectory initially has a larger cost-to-go, its states are not selected as terminal candidates, causing the controller to refine only the above path and converge to a high-cost local optimum (Fig.~\ref{fig:example}). 

In contrast, MM-LMPC classifies the initial trajectories into separate modes and maintains a controller for each. 
The LCB-based meta-controller continues to execute the below mode despite its initial suboptimality, gradually reducing its cost. 
Figure~\ref{acc_final_trajectories} compares the final iteration results: while the original LMPC converges to the suboptimal upper path with a final cost of 18, MM-LMPC successfully identifies and exploits the globally better lower path, achieving a lower final cost of 17. 
Moreover, Figure~\ref{acc_mm_exploration} illustrates all trajectories generated during learning, and we can observe that MM-LMPC systematically explores both candidate routes. %and progressively improves the lower path until it emerges as the best solution, demonstrating the ability of the proposed method to avoid entrapment in local optima and to achieve globally superior performance.

%Once the below-path cost falls below that of the above path, MM-LMPC switches to exploiting it and ultimately converges to the globally optimal solution with a cost of 17 steps, outperforming standard LMPC (18 steps).
 %This result validates our claims and demonstrates the efficacy of MM-LMPC in overcoming the limitations of standard LMPC in multi-modal environments.

\section{Conclusion}

In this paper, we proposed Multi-Modal Learning Model Predictive Control (MM-LMPC), a framework that mitigates the tendency of standard LMPC to converge to high-cost local optima by maintaining mode-specific controllers coordinated by a bandit-based meta-controller. We showed that MM-LMPC preserves the recursive feasibility and stability, while providing convergence within each mode and a logarithmic regret bound on its exploration process. A simulation study on a Dubins car problem demonstrated that, unlike standard LMPC, which remained confined to a single mode, MM-LMPC was able to improve multiple modes in parallel and achieve lower costs. %These results highlight its potential to enhance the robustness and applicability of learning-based control in complex iterative tasks. 

In future work, we plan to demonstrate the utility of the proposed framework in more challenging experiments involving a larger number of modes and richer task structures. We also aim to relax some of the simplifying assumptions adopted in the analysis, thereby extending the theoretical guarantees of MM-LMPC to broader settings.

\bibliographystyle{IEEEtran}

\bibliography{IEEE}

\end{document}